\documentclass[letterpaper,11pt]{article}
\usepackage{amssymb}

\usepackage{amsmath,amssymb,amsthm}
\usepackage{algorithm,setspace}
\usepackage{algorithmic}
\usepackage[dvips]{graphicx}
\usepackage[letterpaper,nohead,left=0.8in,right=0.8 in,top=1.25in,bottom=1.25in ]{geometry}
\spacing{1.1}

\def\setof#1{\left\{{\let\st\colon #1 }\right\}}

  \def\calG{\mathcal{G}} 
   
 \def\calT{\mathcal{T}}
 \def\Reals#1{\mathbb{R}^{#1}}

 \def\BOO={=^{\hspace{0.06cm}\epsilon}_B}
  \def\calM{\mathcal{M}}
  \def\00{\mathbf{0}}
\def\aa{\mathbf{a}}   
\def\uu{\mathbf{u}} \def\vv{\mathbf{v}} \def\xx{\mathbf{x}} \def\yy{\mathbf{y}}
\def\AA{\mathbf{A}} \def\BB{\mathbf{B}}  
\def\pp{\mathbf{p}}   \def\ss{\mathbf{s}}
 \def\xx{\mathbf{x}} \def\yy{\mathbf{y}} 
   
 \def\ww{\mathbf{w}} \def\CC{\mathbf{C}} \def\DD{\mathbf{D}}
 \def\calG{\mathcal{G}}

 \def\calX{\mathcal{X}}  
\newtheorem{coro}{Corollary} \newtheorem{defi}{Definition} 
\newtheorem{lemm}{Lemma}  \newtheorem{theo}{Theorem}
\newtheorem{prope}{Property}  
 
 \def\opt{\text{\sc OPT}}

\begin{document}
\title{\LARGE Settling the Complexity of Arrow-Debreu Equilibria\vspace{0.05cm} \\
   in Markets with Additively Separable Utilities\vspace{1cm}}
\author{
    Xi Chen\hspace{0.05cm}\footnote{Department of Computer Science, Princeton University.}
    \and
    Decheng Dai\hspace{0.05cm}\footnote{Department of Computer Science,
      Tsinghua University. Work done while visiting Princeton University.
      This work was supported in part by the National Natural Science
      Foundation of China Grant 60553001, and the National Basic Research
      Program of China Grant 2007CB807900, 2007CB807901.}
    \and
    Ye Du\hspace{0.05cm}\footnote{Department of Electrical Engineering and Computer Science,
      University of Michigan.}
    \and
    Shang-Hua Teng\hspace{0.05cm}\footnote{Microsoft Research New England.
      Affiliation starting from Fall 2009:
      Department of Computer Science, University of Southern California.}
} \date{} \maketitle

\begin{abstract}
We prove that the problem of computing an Arrow-Debreu market
  equilibrium is PPAD-complete even when
  all traders use additively separable, piecewise-linear and concave utility functions.
In fact, our proof shows that this market-equilibrium problem does
  not have a fully polynomial-time appro\-ximation scheme unless
  every problem in PPAD is solvable in polynomial time.
\end{abstract}

\thispagestyle{empty}
\newpage
\setcounter{page}{1}

\section{Introduction}

One of the central developments in mathematical economics
  is the general equilibrium theory, which provides the foundation
  for competitive pricing \cite{AD,ScarfPrice}. 
When specialized to exchange economies, it con\-siders
  an {\em exchange market} in which there are $m$ traders and
  $n$ divisible goods, where
  trader $i$ has an {\em initial endowment} of
  $w_{i,j} \geq 0$ of good $j$ and a {\em utility function}
  $u_{i}: \Reals{n}_+ \rightarrow \Reals{}$.
The {\em individual goal} of trader $i$ is to
  obtain a new bundle of goods that maximizes her utility.
This new bundle can be specified by a column vector
   $\xx_{i} \in \Reals{n}_+$, where the $j^{th}$ entry $x_{i,j}$
   is the amount of good $j$ that trader $i$ is able to obtain after
   the exchange.
Naturally, the exchange should satisfy
  {$\sum_{i} x_{i,j} \leq \sum_{i} w_{i,j},$} for all good $j$.\vspace{0.02cm}

The pioneering equilibrium theorem of Arrow and Debreu \cite{AD}
  states that if all the utility functions $u_1,...,u_m$ are quasi-concave, then
  under some mild conditions, the market has an {\em equilibrium price}
  $\pp = (p_{1},..., p_{m})$:
At this price, independently, each trader can sell her endowment virtually
  to the market to obtain a budget and then
  buys a bundle of goods with this budget from the market --- which
  contains the union of all goods --- that maximizes her utility.
The equilibrium condition guarantees that
  the supply equals the demand and hence the market
  clears: Every good is sold and every trader's budget is completely
  spent.
In the case when the utility functions are strictly concave,
  there is a unique optimal bundle of goods for each trader at any
  given price $\pp$.
Nevertheless, the theorem extends to quasi-concave
  utility functions such as linear or~piece\-wise linear utility
  functions \cite{Maxfield, Geana},
even though they are not
  strictly quasi-concave, and there could be multiple
  optimal bundles of goods for each trader at a given price.\vspace{0.01cm}

The existence proof of Arrow and Debreu \cite{AD},
  based on Kakutani's fixed point theorem \cite{KAKU},
  is~non-constructive in the view of
  polynomial-time computability.
Despite the progress both on algorithms~for and on the
complexity-theoretic
  understanding of market equilibria,
 several fundamental questions con\-cern\-ing market equilibria,
 including some seemingly simple ones, remain unsettled.\vspace{0.02cm}

Vijay Vazirani \cite{AGTBook} wrote:\vspace{0.03cm}
\begin{quote}
{\em ``Concave utility functions, even if they are additively
separable over the goods, are not easy \\to deal with algorithmically.
In fact, obtaining a polynomial time algorithm for such functions\\ is
a premier open question today.''}\vspace{0.03cm}
\end{quote}
A function $u (x_{1},...,x_n)$ from $\mathbb{R}_+^n$ to $\mathbb{R}$
  is an {\em additively separable} and
  {\em concave} {function}  if there exist real-valued concave functions
  $f_{1},...,f_{n}$ such that
$$u (x_{1},...,x_{n}) = \sum_{j=1}^n f_j(x_j).$$
Noting that every concave function can be approximated by a
  piecewise linear and concave (PLC) function,
  Vazirani \cite{AGTBook} further
  asked whether one can compute a market
  equilibrium with additively separable PLC utilities
  in polynomial time; or whether the problem is PPAD-hard.
This open question has been echoed in several work
  since then \cite{Nikhil,HuangTeng,AuctionVazirani,VaziManu}.

\subsection{Our Contribution}\label{}

In this paper, we settle the complexity of finding
  an Arrow-Debreu equilibrium in an exchange market with addi\-tively separable PLC utilities.
We show that this equilibrium problem is PPAD-complete.\vspace{0.01cm}

For an integer $t > 0$, a real-valued function $f(\cdot)$ is
$t$-segment
  piecewise linear over $\mathbb{R}_+ =[0,+\infty)$ if $f$ is continuous
  and $\mathbb{R}_+$ can be divided
  into $t$ sub-intervals such that $f$ is a linear function over every
  sub-interval.
If each trader's utility  is an additively separable $t$-segment PLC
  function, then we refer to the market as a {\em $t$-linear market}.
Clearly, a market with linear utilities is a $1$-linear market. In
contrast to the fact that an Arrow-Debreu market equilibrium of a
  $1$-linear market can be found in polynomial time
  \cite{EG,Nenakov,DengSafra,DevanurPapa,JainAD}, we show that
  even computing an Arrow-Debreu equilibrium in a $2$-linear
   market is PPAD-complete, via a reduction from
   {\sc Sparse Bimatrix} \cite{ChenDengTengSparse}:
   the problem of finding an approximate Nash equilibrium
   in a sparse two-player game (see Section \ref{sparsedefinition} for
   the definition).\vspace{0.01cm}

Our construction of the PPAD-complete markets has several
  nice technical elements.
First we introduce a sequence of simple \emph{linear} markets
  $\{\calM_n\}$ with $n$ goods,
   which we refer to as the \emph{price-regulating} markets.
$\calM_n$ has the following nice \emph{price-regulation property}:
If $\pp$ is a \emph{normalized}\hspace{0.05cm}\footnote{We say a price
vector $\pp$  is normalized if the smallest nonzero entry of $\pp$ is
  equal to $1$.} approximate equilibrium price vector of $\calM_n$,
  then $p_k\in [1,2]$ for all $k\in [n]$.
This price-regulation property
   allows us to encode $n$ free variables $x_1,...,x_n$ between $0$ and $1$ using
  the $n$ entries of $\pp$ by setting $x_k=p_k-1$.\vspace{0.015cm}

As a key step in our analysis, we show that the price-regulation property is
\emph{stable} with respect to
  ``\emph{small perturbations}'' to $\calM_n$:
  When new traders are added to $\calM_{n}$ (without introducing
  new goods),
  this property remains hold as long as the amount of goods these
  traders carry with them is small
  compared to those of the traders in $\calM_n$.
We then show how to set the initial endowments and utility functions of new traders
  so that we can control the flows of goods in the market and
  set new requirements that every approximate equilibrium price vector
  $\pp$ has to satisfy.\vspace{0.018cm}

Using the stability of the price-regulating market, we give a reduction from
  a two-player game to~an exchange market $\calM$:
Given an $n\times n$ two-player game $(\AA,\BB)$, we construct
  an additively separable PLC market by adding new traders
  --- whose initial
  endowments are relatively small --- to $\calM_{2n+2}$, the price-regulating market
  with $2n+2$ goods.
We use the first $2n$ entries of $\pp$ to encode a pair of
probability vectors
  $(\xx,\yy)$: $x_k=p_k-1$ and $y_k=p_{n+k}-1$, $k\in [n]$.
We then develop a novel way to enforce the Nash equilibrium
  constraints over $\AA,\hspace{0.04cm}\BB,\hspace{0.04cm}\xx$ and
  $\yy$  by carefully specifying the behaviors of the new traders.
In doing so, we construct a market $\calM$ with the property that
  from every approximate market
  equilibrium $\pp$ of $\calM$, the pair $(\xx,\yy)$ obtained above
  (after normalization)
  is an approximate Nash equilibrium of $(\AA,\BB)$.
Moreover, if $(\AA,\BB )$ is a sparse two-player game, then
  the relation of which traders are interested in which goods in
  $\calM$ is also sparse (see Section \ref{lulala} for details).\vspace{0.01cm}

In the construction of $\calM$, the price-regulation property
   plays a critical role.
It enables  us to design the utility functions of new traders
  so that we know exactly their preferences over the goods with respect
  to any approximate equilibrium price $\pp$, even though we have no
  idea in advance about the entries of $\pp$ when constructing $\calM$.\vspace{0.014cm}

We anticipate that our reduction techniques will help to resolve
  more complexity-theoretic questions concerning other families of exchange markets
  such as the general CES markets
  and the hybrid linear-Leontief mar\-kets \cite{ChenHuangTengWine}.\newpage

\subsection{Related Work}\label{}

The computation of a market equilibrium price in an exchange market
  has been a challenging problem in mathematical economics \cite{Scarf,AGTBook}.
The matter is more complex because some markets only have irrational~eq\-uilibria,
  making the computation of exact
  equilibria with a finite-precision algorithm impossible.
One alternative approach to handle irrationality is to express
  equilibria in some simple algebraic form.
However, it turns out that finding an exact
  market equilibrium in general is not computable \cite{WongRichter}.\vspace{0.015cm}

To circumvent the irrationality, one usually uses some notion of
approximate
   market equilibria.
There are various notions of approximate equilibria: Some require
that the approximation solution is within a small
  geometric distance from an exact equilibrium, while
  others only require that the supply-demand condition and/or
  the individual optimality condition are approximately satisfied.
In this paper, following Scarf \cite{Scarf},
   we consider the latter notion of approximate market equilibria.

\subsubsection{Algorithms for Market Equilibria}\label{}

Scarf pioneered the algorithmic development for computing general
competitive
  equilibria \cite{Scarf}.
His approach combined numerical approximation with combinatorial
  insights used in Sperner's lemma \cite{SPE28} for fixed points
  and in Lemke and Howson's algorithm for two-player games.
Although his algorithm may not always run in polynomial time,
  Scarf's work has profound impact to computational economics.\vspace{0.015cm}

Building on the success of convex programming \cite{EG},
  polynomial-time algorithms have been developed for special markets
  whose sets of equilibria enjoy some degree of convexity.
For Arrow-Debreu markets with~linear utility functions, Nenakov and Primak
  gave a polynomial-time algorithm \cite{Nenakov},
  and there are now several polynomial-time algorithms for computing or
  approximating market equilibria with linear utility functions
  \cite{DengSafra,DevanurPapa,JainAD,KapoorGarg,JainMahdianSaberi,DeVazirani,YeArrow}.
Other polynomial-time algorithms for special
  markets include Eaves's algorithm for Cobb-Douglas markets
  \cite{Eaves} and Devanur and Vazirani's algorithm for markets~with spending
   constraint utilities~\cite{SpendingConstraint}
   (also see \cite{VaziManu}).
The convex programming based approach for approximating equilibria has
  been extended to all markets whose utilities satisfy
  weak gross substitutability (WGS) by Codenotti, Pemmaraju, and Varadarajan \cite{GeneralWGS}.
In \cite{CodeMcCune}, Codenotti, McCune, and Varadarajan showed that
  for markets that satisfy WGS, there is a price-adjustment mechanism
  called \emph{t$\hat{a}$onnement} that converges to an
  approximate equilibrium efficiently.\vspace{0.02cm}

A closely related market model is
   Fisher's model \cite{BandScarf}.
In this model, there are two different types~of traders in the market: {\em producers} and
  {\em consumers}.
Each consumer comes to the market with a budget and
  a utility~fun\-ction.
Each producer comes to the market with an endowment of goods, and
  will sell them to the con\-sumers for money.
A market equilibrium is then a price vector $\pp$ for goods so that if each
 consumer spends all her budget to maximize
 her utility, then the market clears.
An (approximate) market equilibrium in a Fisher's market with
  CES\hspace{0.01cm}\footnote{CES (standing for constant
 elasticity of substitution) is a  popular family of utility functions.
Let $s>0$ be a parameter called the \emph{elasticity of substitution},
  then a CES function with elasticity of substitution $s$ has the following form:
\[
 u (x_{1},...,x_{m})
   = \left(\sum_{j=1}^{m} d_{j}x_{j}^{r } \right)^{1/r },\ \ \ \ \text{where $r=
   \frac{s-1}{s}$}.
 \]
 } utility functions \cite{EG,YeArrow,YeRationality,DevanurPapa,JainVaziraniYe}
  or with piecewise\newpage \noindent linear utility functions \cite{YeRationality}
  can be found in polynomial time.
In \cite{CDSY}, Chen, Deng, Sun, and Yao gave an algorithm for markets
  with logarithmic utility functions.
Its running time is polynomial when either the number of sellers
  or the number of buyers is bounded by a constant.\vspace{0.018cm}

However, progress on Arrow-Debreu markets whose sets of equilibria do not enjoy
 convexity has been relatively slow.
There are only a few algorithms in this category. Devanur and Kannan
\cite{Nikhil} gave~a~poly\-nomial-time algorithm for exchange markets
  with PLC utility functions and a constant number of goods.
Codenotti,  McCune, Penumatcha, and Varadarajan
gave a polynomial-time
  algorithm for CES markets when the elasticity of substitution $s\ge 1/2$ \cite{Codenotti}.

\subsubsection{The Complexity of Equilibrium Problems}\label{}

Papadimitriou initiated the complexity-theoretic study of
  fixed-point computations \cite{PAP94}.
He introduced the complexity class PPAD, and proved that
  the problem of finding a Nash equilibrium
  in a two-player game, the computational version of Sperner's Lemma, and
  the problem of computing an approximate fixed point
  are members of PPAD.\vspace{0.01cm}

Recently, there was a series of developments that characterized
  the computational complexity of se\-veral equilibrium problems in game
  theory.
Daskalakis, Goldberg, and Papadimitriou \cite{GOL05} proved
that
  the problem of computing an exponentially-precise Nash equilibrium
  of a four-player game is PPAD-complete.
Chen and Deng \cite{CHE06} then proved that the problem of
  computing a two-player Nash equilibrium is also PPAD-complete.
Chen and Deng's result, together with an earlier reduction of
\cite{CSVY},
  implies that computing a market equilibrium in an Arrow-Debreu market
  with Leontief utilities\hspace{0.05cm}\footnote{Leontief
    functions are special cases of CES functions with $s$ approaching 0.
A Leontief function has the following form:
$u(x_{1},...,x_{m}) = \min_{j\in S} x_{j}/d_{j}$, where $S\subseteq [m]$ is a subset of goods
  and $d_j>0$ for all $j\in S$.} is PPAD-hard.~On
the approximation front, Chen, Deng and Teng \cite{ChenDengTengFOCS} proved
  that it is PPAD-complete to find a  polyno\-mially-precise
  approximate Nash equilibrium in two-player or multi-player games.
Huang and Teng \cite{HuangTeng} then extended this approximation
  result to Leontief market equilibria.
Their approximation result also implies that the
  market equilibrium problem with CES utility
  functions is PPAD-hard,
  if the elasticity of substitution $s$ is sufficiently small.

\section{Preliminaries}

\subsection{Complexity of Nash Equilibria in Sparse Two-Player Games}\label{sparsedefinition}

A two-player game is defined by the payoff matrices $(\AA,\BB)$ of its
  two players.
In this  paper, we assume that both players have $n$ choices of actions
  and hence both $\AA$ and $\BB$ are square matrices  with $n$ rows and columns.
We use $\Delta^n\subset \mathbb{R}^n$ to denote the set of probability distributions
  of $n$ dimensions.
A pair of probability vectors $(\xx,\yy)$ (i.e., $\xx\in \Delta^n $
  and $\yy\in \Delta^n$)
  is a Nash equilibrium of $(\AA,\BB)$ if for all $i$
  and $j$ in $[n]=\{1,2,...,n\}$,
$\AA_i\yy^T <\AA_j\yy^T\Longrightarrow x_i=0$ and
$\xx\BB_i <\xx\BB_j\Longrightarrow y_i=0,$
where we use $\AA_i$ and $\BB_i$ to denote the $i$th row vector of $\AA$
  and the $i$th column vector of $\BB$, respectively.
We will use the following notion of approximate Nash equilibria.\newpage

\begin{defi}[Well-Supported Nash Equilibria]\label{hohoho}
For $\epsilon>0$,  $(\xx,\yy)$ is an
  $\epsilon$-well-supported Nash equilibrium of $(\AA,\BB)$,
  if $\xx,\yy\in \Delta^n$ and for all $i,j\in [n]$,\vspace{-0.1cm}
\begin{eqnarray}\label{nasheq1}
&\AA_i\yy^T+\epsilon<\AA_j\yy^T\ \ \Longrightarrow\ \ x_i=0,& \text{and}\\
\label{nasheq2}
&\xx\BB_i+\epsilon<\xx\BB_j \ \ \Longrightarrow\ \ y_i=0.
\end{eqnarray}
\end{defi}

\begin{defi}[Sparse Normalized Two-Player Games]\label{sparsehuhu}
A two-player game $(\AA,\BB)$ is \emph{normalized} if every entry of
  $\AA$ and $\BB$ is between $-1$ and $1$.
We say a two-player game $(\AA,\BB)$ is \emph{sparse} if
  every row and every column of $\AA$ and $\BB$ have at most $10$ nonzero entries.\vspace{0.06cm}
\end{defi}

Let {\sc Sparse Bimatrix} denote the problem of finding an
$n^{-6}$-well-supported
  Nash equilibrium in an $n\times n$ sparse normalized
  two-player game, then we have\vspace{0.06cm}

\begin{theo}[Chen-Deng-Teng~\cite{ChenDengTengSparse}]
{\sc Sparse Bimatrix} is \emph{PPAD}-complete.\vspace{0.06cm}
\end{theo}

\subsection{Markets with Additively Separable PLC Utilities}\label{mainresultsec}

Let $\calG=\{G_1,..., G_n\}$ denote a set of $n$ divisible goods,
  and $\calT=\{T_1,...,T_m\}$ denote a set of $m$ traders.
For each trader $T_i\in \calT$, we use $\ww_i\in \mathbb{R}_+^n$
  to denote her initial endowment, $u_i:\mathbb{R}_+^n\rightarrow
  \mathbb{R}_+$ to denote her utility function,
  and $\xx_i\in \mathbb{R}_+^N$ to denote her allocation vector.
In this paper, we will focus on markets with additively
  separable piecewise linear and concave utilities.\vspace{0.08cm}

\begin{defi}
A function $r(\cdot):\mathbb{R}_+\rightarrow \mathbb{R}_+$ is said to be
  \emph{$t$-segment}
  \emph{piecewise linear and concave} \emph{(PLC)} if
\begin{enumerate}
\item $r(0)=0$ and $r(\cdot)$ is continuous over $\mathbb{R}_+$;\vspace{-0.08cm}
\item there exists a tuple of length $2t+1$, $\big[\theta_0>\theta_1>...>\theta_t;
  a_1<a_2<...<a_t\big]\in \mathbb{R}_+^{2t+1},$ such that\vspace{-0.08cm}
\begin{enumerate}
\item  for any $i\in [0:t-1]$, the restriction of $f$ over $[a_i,a_{i+1}]$
  \emph{(}$a_0=0$\emph{)} is a segment of slope $\theta_i$;
\item the restriction of $f$ over $[a_t,+\infty)$ is a ray of slope $\theta_t$.
\end{enumerate}
\end{enumerate}
The $2t+1$-tuple $[\theta_0,\theta_1,...,\theta_t;a_1,a_2,...,a_t]$ is also called
  the \emph{representation} of $r(\cdot)$.
Moreover, we say $r(\cdot)$ is \emph{strictly monotone} if $\theta_t>0$, and
  is $\alpha$-\emph{bounded}
  for some $\alpha\ge 1$ if\vspace{-0.1cm}
$$\alpha\ge \theta_0>\theta_1>...>\theta_t\ge 1.$$
\end{defi}

\begin{defi}
A utility function $u(\cdot):\mathbb{R}_+^n\rightarrow \mathbb{R}_+$
  is said to be an \emph{additively separable PLC} function if there exist a set of
  $n$ \emph{PLC} functions $r_{1}(\cdot),...,r_{n}(\cdot):\mathbb{R}_+\rightarrow
  \mathbb{R}_+$ such that
\begin{equation}\label{additive}
u(\xx)=\sum_{j\in [n]} r_{j}(x_j),\ \ \ \text{for all $\xx\in \mathbb{R}_+^n$.}
\end{equation}
\end{defi}

In such a market, we use, for each trader $T_i\in \calT$,
  $r_{i,j}(\cdot):\mathbb{R}_+\rightarrow \mathbb{R}_+$ to denote her
  PLC function with respect to good $G_j\in \calG$.
In another word, we have
$$
u_i(\xx)=\sum_{j\in [n]} r_{i,j}(x_j),\ \ \ \text{for all $\xx\in \mathbb{R}_+^n$.}
$$

We use $\pp\in \mathbb{R}_+^n$ to denote a price vector, where $\pp\ne \mathbf{0}$ and
  $p_j$ is the price of $G_j$, $j\in [n]$.
We always assume that $\pp$ is \emph{normalized}, that is,
  the smallest nonzero entry of $\pp$ is equal to $1$.\vspace{0.015cm}

Given $\pp$, we use $\opt(i,\pp)\subset\mathbb{R}_+^n$ to denote the
  set of allocations that maximize the utility of $T_i$:
\begin{equation*}
\opt(i,\pp)=\text{argmax}_{\hspace{0.08cm}\xx\in \mathbb{R}_+^n,\
  \xx\cdot \pp\le \ww_i\cdot \pp}\  u_i(\xx).
\end{equation*}
We use $\mathcal{X}=\{\xx_i\in \mathbb{R}_+^n:i\in [m]\}$ to denote an allocation of the market:
  For each trader $T_i\in \calT$,
  $\xx_i \in \mathbb{R}_+^n$ is the amount of goods that $T_i$ receives.
In particular, the amount of $G_j$ that $T_i$ receives in $\mathcal{X}$ is
  $x_{i,j}$.\vspace{0.06cm}

\begin{defi}[Arrow-Debreu~\cite{AD}]
A market equilibrium is a \emph{(normalized)} price vector $\pp\in \mathbb{R}_+^n$ such that
  there exists an allocation $\calX$ which has the following properties:
\begin{enumerate}
\item The market clears: For every good $G_j\in \calG$,
\begin{equation}\label{marketclears1}
\sum_{i\in [m]} x_{i,j}\le \sum_{i\in [m]} w_{i,j}.
\end{equation}
In particular, if $p_j>0$, then\vspace{-0.15cm}
\begin{equation}\label{marketclears2}
\sum_{i\in [m]} x_{i,j}= \sum_{i\in [m]} w_{i,j}.
\end{equation}
\item Every trader gets an optimal bundle: For every $T_i\in \calT$,
  we have $\xx_i\in \text{\sc OPT}(i,\pp).$\vspace{0.08cm}
\end{enumerate}
\end{defi}

In general, not every market has such an equilibrium price vector.
For the additively separable PLC markets considered here, the
  following condition guarantees the existence of an equilibrium.\vspace{0.08cm}

\begin{defi}[Economy Graphs]
Given an additively separable \emph{PLC} market, we build a directed graph
  $G=(\calT,E)$ as follows.
The vertex set of $G$ is exactly $\calT$, the set of traders in the market. For every
  two traders $T_i\ne T_j\in \calT$, we have an edge from $T_i$ to
  $T_j$ if there exists an integer $k\in [n]$ such that $w_{i,k}>0$ and
  $r_{j,k}(\cdot)$ is strictly monotone.
In another word, $T_i$ possesses a good which $T_j$ wants.
$G$ is called the \emph{economy graph} of the market \emph{\cite{Maxfield,Codenotti}}.
We say the market is \emph{strongly connected} if $G$ is strongly connected.\vspace{0.08cm}
\end{defi}

The following theorem is a corollary of Maxfield \cite{Maxfield},
  and the proof can be found in Appendix~\ref{appendix1}.\vspace{0.08cm}

\begin{theo}\label{existence}
Let $\calM$ be a market with additively separable \emph{PLC} utilities.
If it is strongly connected, then
  a market equilibrium $\pp$ exists. Moreover, if all the parameters
  of $\calM$ are rational numbers, then it has a rational
  market equilibrium $\pp$.
The number of bits we need to describe $\pp$
  is polynomial in the input size of $\calM$ \emph{(}that is, the number of bits we need
  to describe the market $\calM$\emph{)}.\vspace{0.06cm}
\end{theo}

\subsection{Definition of the Sparse Market Equilibrium Problem}\label{lulala}

By Theorem \ref{existence}, the following search problem {\sc Market} is well defined:
\begin{quote}
The input of the problem is an additively separable PLC market $\calM$
  that is both rational\\ and strongly connected; and the output is a
  rational market equilibrium $\pp$ of $\calM$.
\end{quote}
In the rest of the section, we define a much more restricted version of {\sc Market}:
  {\sc Sparse Market}.
The main result of the paper is that {\sc Sparse Market} is PPAD-complete.

First of all, the input of {\sc Sparse Market} is an additively separable
  PLC market which not only is strongly connected,
  but also satisfies the following three conditions:\vspace{0.07cm}

\begin{defi}[$\alpha$-Bounded Markets]
We say an additively separable \emph{PLC} market $\calM$ is \emph{$\alpha$-bounded},
  for some $\alpha\ge 1$, if for all $T_i$ and $G_j$,
  $r_{i,j}(\cdot)$ is either the zero function \emph{(}$r_{i,j}(x)=0$ for all
  $x$\emph{)} or $\alpha$-bounded.\vspace{0.07cm}
\end{defi}

\begin{defi}[$2$-Linear Markets]
We call an additively separable \emph{PLC} market $\calM$
  a \emph{$2$-linear} market if for all
  $T_i\in \calT$ and $G_j\in \calG$,
  $r_{i,j}(\cdot)$ has at most two segments.\vspace{0.07cm}
\end{defi}

\begin{defi}[$t$-Sparse Markets]
We say an additively separable \emph{PLC} market $\calM$ is
  $t$-sparse for some integer $t>0$ if \emph{1)}
  For every $T_i\in \calT$, we have $|\hspace{0.04cm}
  {\text{\emph{supp}}}(\ww_i)|\le t;$ and \emph{2)}
For every $T_i\in \calT$, the number of $j\in [n]$ such that
  $r_{i,j}(\cdot)$ is not the zero function is at most $t$. In
  another word, every trader owns at most
  $t$ goods at the beginning and is interested in at most $t$ goods.\vspace{0.04cm}
\end{defi}

We use the following definition of approximate market equilibria:\vspace{0.04cm}

\begin{defi}[$\epsilon$-Approximate Market Equilibrium]
Given an additively separable \emph{PLC} market $\calM$, we say $\pp$ is
  an $\epsilon$-approximate market equilibrium of $\calM$,
  for some $\epsilon\ge 0$, if there is an allocation $\calX=\{\xx_i\in
  \mathbb{R}_+^n:i\in [m]\}$ such that every trader gets an optimal bundle
  with respect to $\pp$:
$\xx_i\in \opt(i,\pp)$ for all $i \in$ $ [m]$;
and the market clears approximately: For every $G_j\in \calG$,
\begin{equation}\label{marketclears2}
\left|\sum_{i\in [m]} x_{i,j}-\sum_{i\in [m]} w_{i,j}\right|
\le \epsilon\cdot \sum_{i\in [m]} w_{i,j}.\vspace{0.04cm}
\end{equation}
\end{defi}

We remark that there are various notions of approximate market equilibria.
The reason we adopted the one above
  is to simplify the analysis.
The construction in Section \ref{allconstruction} actually works for some other notions
  of approximate equilibria,
  e.g., the one that also allows the allocation to be just approximately optimal for
  each trader.

Finally, we let {\sc Sparse Market} denote the following search problem:
\begin{quote}
The input of the problem is a 2-linear market $\calM$ that is
  strongly connected, $27$-bounded, \\ and $23$-sparse; and the output is an
  $n^{-13}$-approximate market equilibrium of $\calM$, where $n$ \\is the number of
  goods in the market.
\end{quote}
It is tedious but not very hard to show that {\sc Sparse Market} is
  a problem in PPAD\hspace{0.05cm}\footnote{
  In \cite{Geana}, the author showed how to construct a continuous map
    from any market with quasi-concave utilities such that the set of
    fixed points of the map is precisely the set of equilibria of the market.
  When the market is additively separable PLC, one can show that
    the continuous map is indeed Lipschitz continuous.
  As a result, one can reduce the problem of finding an approximate market
    equilibrium to the problem of finding an approximate fixed point in
    a Lipschitz continuous map.
  This implies a reduction from {\sc Sparse Market} to the discrete
    fixed point problem studied in \cite{HPV} (also see \cite{ChenDengTengFOCS}
    for the high-dimensional version) which is in PPAD, and thus,
    the former is also in PPAD.}.\newpage

One can in fact replace the constant
  $27$ here by any constant larger than $1$ and our main result,
  Theorem \ref{main}, below still holds.
The constant $23$, however, is related to the
  constant $10$ in Definition \ref{sparsehuhu}.

The main result of the paper is the following theorem:
\begin{theo}[Main]\label{main}
\emph{{\sc Sparse Market}} is \emph{PPAD}-complete.\vspace{0.04cm}
\end{theo}

\section{A Price-Regulating Market}\label{priceregulating}

We now construct the family of price-regulating market $\{\calM_n\}$.
For each positive integer $n\geq 2$, $\calM_n$ has $n$ goods and satisfies
  the following strong price regulation property.\vspace{0.05cm}

\begin{prope}[Price Regulation]\label{simplelemma}
A price vector $\pp$ is a normalized $n^{-1}$-approximate equilibrium
  of $\calM_n$  if and only if $
1\le p_k\le 2,$ for all $k\in [n]$.\vspace{0.07cm}
\end{prope}

We start with some notation.
The goods in $\calM_n$ are  $\calG=\{G_1,...,G_n\},$ and
 the traders in $\calM_n$ are
$$
\calT=\Big\{T_\ss: \ss\in S\Big\},
\ \ \ \ \text{where\ \ \ \ $S=\Big\{\ss=(i,j): 1\le i\ne j\le n \Big\}$}.
$$
For every trader $T_\ss\in \calT$, we use $\ww_\ss\in \mathbb{R}_+^n$ to denote
  her initial endowment, $u_\ss:\mathbb{R}_+^n\rightarrow \mathbb{R}_+$
  to denote her utility function, $r_{\ss,k}(\cdot)$ to denote her
  PLC function with respect to $G_k$, and $\opt(\ss,\pp)$ to denote
  the set of bundles that maximize her utility with respect to $\pp$.\vspace{0.01cm}

Market $\calM_n$ is a linear market in which for all $\ss\in S$ and
  $k\in [n]$, $r_{\ss,k}(\cdot)$ is a ray starting at $(0,0)$.
In the construction below, we let $r_{\ss,k}(\cdot)\Leftarrow [\theta]$
  denote the action of setting $r_{\ss,k}(\cdot)$
  to be the linear function of slope\vspace{-0.2cm} $\theta\ge 0$.\newline

\noindent\textbf{Construction of $\calM_n$:}

First, we set the initial endowment vectors $\ww_\ss$:
  For every $\ss=(i,j)\in S$, we set $w_{\ss,k}= 1/n$ if $k=i$; and
  $w_{\ss,k}=0$ otherwise.

Second, we set the PLC functions $r_{\ss,k}(\cdot)$:
  For all $\ss=(i,j)\in S$ and $k\in [n]$, we set
  $r_{\ss,k}(\cdot)\Leftarrow [\theta]$ and
  $\theta=0$ if $k\ne i,j$; $\theta=1$ if $k=j$;
  and $\theta=2$ if $k=i$.

It is easy to check that $\calM_n$ constructed above is strongly connected, $2$-bounded,
  and\vspace{-0.35cm} $2$-sparse.\newline

\begin{proof}[Proof of Property \ref{simplelemma}]
The first direction is trivial. If $1\le p_k\le 2$ for all $k\in [n]$,
  then one can verify\vspace{-0.06cm} that
$$\calX=\big\{\xx_\ss=\ww_\ss:\ss\in S\big\}\vspace{-0.06cm}$$
  is a market clearing allocation that provides an optimal bundle of
  goods for each trader at price $\pp$.

The second direction is less trivial. Let $\pp$ be a normalized $(1/n)$-approximate
  market equilibrium
  of $\calM_n$, and $\calX$ be an optimal allocation that clears the market.
First, it is easy to check that $p_k$ must be positive for all $k\in [n]$ since
  otherwise, we have $x_{\ss,k}=+\infty$ for all $\ss=(i,j)$ such that $k=i$ or $j$,
  which contradicts the assumption that $\pp$ is an approximate equilibrium.\vspace{0.01cm}

Since $\pp$ is normalized, we have $p_k\ge 1$ for all $k\in [n]$.
Now assume for contradiction that Property \ref{simplelemma} is not true, then
  without loss of generality, we may assume that $p_1=\max_k p_k >2$
  and $p_2=\min_kp_k=1$.
To reach a contradiction, we focus on the amount of $G_1$ each
  trader gets in the allocation $\calX$.
First, if $1\notin \{i,j\}$ where $\ss=(i,j)$, then 
  we have $x_{\ss,1}=0$;
Second, if $i=1$ and $j=2$, then $x_{\ss,1}=0$ since
  $$
    \frac{2}{p_1}<\frac{1}{p_2}
  $$
and $T_\ss$ likes $G_2$ better than $G_1$ with respect to the price vector $\pp$;
Third, if $j=1$, then $x_{\ss,1}=0$ since
$$
\frac{1}{p_1}<\frac{2}{p_i}
$$
and $T_\ss$ likes $G_i$ better than $G_1$;
Finally, for all $\ss=(i,j)$ such that $i=1$ and $j\ne 2$,
  we have $x_{\ss,1}\le 1/n$ since the budget of $T_\ss$ is exactly $(1/n)\cdot p_1$.
As a result, we have
$$
\sum_{\ss\in S} x_{\ss,1}\le \frac{n-2}{n},\ \ \ \text{while}\ \ \
\sum_{\ss\in S} w_{\ss,1}=\frac{n-1}{n},
$$
which contradicts the assumption that $\pp$ is a $(1/n)$-approximate equilibrium since
$$
\left|\frac{n-2}{n}-\frac{n-1}{n}\right|
  > \frac{1}{n}\cdot \frac{n-1}{n}.
$$
The price-regulation property then follows.
\end{proof}

Let $x_k=p_k-1$ for $k\in [n]$, then $\calM_n$ provides us
  a way to encode $n$ free variables $x_1,...,x_n$ between $0$ and $1$.
In the next section, we will use $\calM_{2n+2}$ and the first $2n$ entries of $\pp$:
  $$x_k=p_k-1\ \ \ \ \text{and}\ \ \ \ y_k=p_{n+k}-1,\ \ \ \ \text{for $k\in [n]$,}$$
to encode a pair of distributions $(\xx,\yy)$.
Starting from an $n\times n$ sparse two-player game $(\AA,\BB)$,
  we will show how to add more traders to ``perturb'' the
  price-regulating market $\calM_{2n+2}$
  so that any approximate equilibrium $\pp$ of the new market yields
  an approximate Nash equilibrium $(\xx,\yy)$ of $(\AA,\BB)$.\vspace{0.06cm}

\section{Reduction from {\sc Sparse Bimatrix} to {\sc Sparse Market}}\label{allconstruction}

In this section, we give a polynomial-time reduction from {\sc Sparse Bimatrix} to
  {\sc Sparse Market}.
Given an $n\times n$ sparse two-player game $(\AA,\BB)$,
  where $\AA,\BB\in [-1,1]^{n\times n}$,
  we construct an additively separable PLC market $\calM$ by
  adding more traders to the price-regulating market $\calM_{2n+2}$.
There are $2n+2$ goods $\calG=\{G_1,...,G_{2n},G_{2n+1},G_{2n+2}\}$ in
  $\calM$, and the traders $\calT$ in $\calM$ are
$$
\calT= \Big\{T_\ss,T_\uu,T_\vv,T_i:\ss\in S,\uu\in U, \vv\in V,i\in [2n]\Big\},$$
where $S=\big\{\ss=(i,j):1\le i\ne j\le 2n+2\big\}$,\vspace{-0.08cm}
\begin{equation*}
U=\big\{\uu=(i,j,1):1\le i\ne j\le n\big\}\ \ \ \text{and}\ \ \
V=\big\{\vv=(i,j,2):1\le i\ne j\le n\big\}.
\end{equation*}
Note that $|\calT|=O(n^2)$.
The traders $T_\ss$, where $\ss\in S$, have almost the same initial endowments $\ww_\ss$
  and PLC functions $r_{\ss,k}(\cdot)$ as in $\calM_{2n+2}$;
  we will only slightly modify these parameters to ease the analysis in the next section.



For each agent $T\in \calT$, we will set her PLC function $r(\cdot)$
  with respect to $G_k$, $k\in [2n+2]$, to one of the following functions:
\begin{enumerate}
\item $r(\cdot)$ is the zero function: $r(x)=0$ for all $x\ge 0$ (denoted by
  $r(\cdot)\Leftarrow [0]$); or
\item $r(\cdot)$ is a ray: $r(x)=\theta\cdot x$ for all $x\ge 0$ (denoted by
  $r(\cdot)\Leftarrow [\theta]$); or
\item $r(\cdot)$ is a 2-segment PLC function with representation
  $[\theta_0,\theta_1;a_1]$ (denoted by
  $r(\cdot)\Leftarrow [\theta_0,\theta_1 ;a_1]$).

\end{enumerate}

\subsection{Setting up the Market}\label{}

For each trader $T\in \calT$, we set her initial
  endowment and PLC utility functions as following:

\subsubsection{Traders $T_\ss$, where $\ss\in S$}

For each trader $T_\ss\in \calT$, where $\ss=(i,j)\in S$,
  we set her initial endowments $\ww_s$ and her
  PLC functions $r_{\ss,k}(\cdot)$  almost the same as hers in $\calM_{2n+2}$.

The initial endowment $\ww_\ss$ is set as:
  $w_{\ss,k}=1/n$ if $k=i$; and $w_{\ss,k}=0$ otherwise, where $k\in [2n+2]$.\vspace{0.01cm}

The PLC functions $r_{\ss,k}(\cdot)$ is set as:
 $r_{\ss,k}(\cdot)\Leftarrow [\theta]$ and
  $\theta=0$ if $k\notin \{i,j\}$; $\theta=1$ if $k=j$;
  and $\theta=2$ if $k=i$, where $k\in [2n+2]$.


\subsubsection{Traders $T_{\uu}$, where $\uu\in U$}

Let $\uu=(i,j,1)$ be a triple in $U$ with $1\le i\ne j\le n$.
We use $\AA_i$ and $\AA_j$ to denote the $i$th and $j$th row
  vectors of $\AA$, respectively.
We define $\CC$ and $\DD$ to be the following $n$-dimensional vectors: For
   $k\in [n]$,
$$(C_k,D_k)=(A_{i,k}-A_{j,k},0)\ \text{if}\ A_{i,k}-A_{j,k}\ge 0;\
\text{and}\ (C_k,D_k)=(0,A_{j,k}-A_{i,k})\ \text{otherwise}.$$
By definition, we have $\CC-\DD=\AA_i-\AA_j$ while both vectors $\CC$ and $\DD$ are
  nonnegative.
Moreover, because $\AA$ is a sparse matrix, the
  number of nonzero entries in either $\CC$ or $\DD$ is at most $20$
  and every entry is between $0$ and $2$.
We also let $E$ and $F$ be the following two nonnegative numbers:\vspace{0.04cm}
$$
(E,F)=\left(\sum_{k\in [n]}D_k-\sum_{k\in [n]}C_k,0\right)\ \text{if}\
  \sum_{k\in [n]}D_k\ge \sum_{k\in [n]} C_k;
\ (E,F)=\left(0,\sum_{k\in [n]}C_k-\sum_{k\in [n]}D_k\right)\ \text{otherwise}.\vspace{0.06cm}
$$
Accordingly, we have $E,F\ge 0$ and
$$
E+\sum_{k\in [n]}C_k=F+\sum_{k\in [n]}D_k.
$$
Moreover, since $\CC$ and $\DD$ are sparse, we also have $$0\le E,F\le
  \max\left(\sum_{k\in [n]} C_k,\sum_{k\in [n]} D_k\right)\le 40.$$

We set the initial endowment vector $\ww_\uu=(w_{\uu,1},...,w_{\uu,2n+1},w_{\uu,2n+2})$
  of $T_\uu$ as follows:
\begin{enumerate}
\item $w_{\uu,i}= 1/n^4$;
  $w_{\uu,k}=w_{\uu,2n+2}=0$ for all other $k\in [n]$;
\item $w_{\uu,n+k}=C_k/n^5$ for all $k\in [n]$; and
\item $w_{\uu,2n+1}=E/n^5$.
\end{enumerate}
It is easy to verify that the number of
  nonzero entries in $\ww_\uu$ is at most $22$.\vspace{0.01cm}

We set the PLC utility functions $r_{\uu,k}(\cdot)$, where
  $k\in [2n+2]$, of $T_\uu$ as follows:
\begin{enumerate}
\item $r_{\uu,i}(\cdot)\Leftarrow [9,1;1/n^4]$; and
  $r_{\uu,k}(\cdot)\Leftarrow [0]$ for all other $k\in [n]$;
\item $r_{\uu,2n+2}(\cdot)\Leftarrow [3]$;
\item $r_{\uu,n+k}(\cdot)\Leftarrow [0]$ for all $k\in [n]$
  such that $D_{k}=0$;
\item $r_{\uu,n+k}(\cdot)\Leftarrow [27,1;D_k/n^5]$ for all $k\in [n]$
  such that $D_k>0$; and
\item $r_{\uu,2n+1}(\cdot)\Leftarrow [0]$ if $F=0$; and
  $r_{\uu,2n+1}(\cdot)\Leftarrow [27,1;F/n^5]$ if $F>0$.
\end{enumerate}
Note that the number of $k\in [2n+2]$
  such that $r_{\uu,k}(\cdot)$ is not the zero function is at most $23$.\vspace{0.015cm}

The constants $1$, $3$, $9$ and $27$ in the construction may look strange at first sight.
The motivation is that, if the price-regulation property still
  holds for the new market $\calM$ (which turns out to be true),
  then we know exactly the preference of $T_\uu$ over the goods since $3>2$.
See the proof of Lemma \ref{hahalemma1} for more details.

\subsubsection{Traders $T_\vv$, where $\vv\in V$}

The behavior of $T_\vv$, $\vv\in V$, is very similar to that of $T_\uu$
  except that it works on the second matrix $\BB$.\vspace{0.01cm}

Let $\vv=(i,j,2)$ be a triple in $V$ with $1\le i\ne j\le n$.
We use $\BB_i$ and $\BB_j$ to denote the $i$th and $j$th column
  vectors of $\BB$, respectively.
Similarly, we define the following $n$-dimensional vectors $\CC$ and $\DD$:
$$
(C_k,D_k)=(B_{k,i}-B_{k,j},0)\ \text{if}\ B_{k,i}-B_{k,j}\ge 0;\
\text{and}\ (C_k,D_k)=(0,B_{k,j}-B_{k,i})\ \text{otherwise}.
$$
As a result, we have $\CC-\DD=\BB_i-\BB_j$ while both $\CC$ and
  $\DD$ are nonnegative.
We also define $E,F\ge 0$ in a similar way so that\vspace{0.04cm}
$$
E+\sum_{k\in [n]} C_k=F+\sum_{k\in [n]} D_k\ \ \ \ \text{and}\ \ \ \ 0\le E,F\le 40.
$$

We set the initial endowment vector $\ww_\vv
  =(w_{\vv,1},...,w_{\vv,2n+1},w_{\vv,2n+2})$ of $T_\vv$ to be
\begin{enumerate}
\item $w_{\vv,n+i}=1/n^4$;
  $w_{\vv,n+k}=w_{\vv,2n+2}=0$ for all other $k\in [n]$;
\item $w_{\vv,k}=C_k/n^5$ for all $k\in [n]$; and
\item $w_{\vv,2n+1}=E/n^5$.
\end{enumerate}
We set the PLC utility functions $r_{\vv,k}(\cdot)$, where $k\in [2n+2]$,
  of $T_\vv$ as follows:
\begin{enumerate}
\item $r_{\vv,n+i}(\cdot)\Leftarrow [9,1;1/n^4]$; and
  $r_{\vv,n+k}(\cdot)\Leftarrow [0]$ for all other $k\in [n]$;
\item $r_{\vv,2n+2}(\cdot)\Leftarrow [3]$;
\item $r_{\vv,k}(\cdot)\Leftarrow [0]$ for all $k\in [n]$
  such that $D_{k}=0$;
\item $r_{\vv,k}(\cdot)\Leftarrow [27,1;D_k/n^5]$ for all $k\in [n]$
  such that $D_{k}>0$; and
\item $r_{\vv,2n+1}(\cdot)\Leftarrow [0]$ if $F=0$; and
  $r_{\vv,2n+1}(\cdot)\Leftarrow [27,1;F/n^5]$ if $F>0$.
\end{enumerate}
Again, the number of nonzero entries in $\ww_\vv$
  is at most $22$, and the number of indices $k$ such that
  $r_{\vv,k}(\cdot)$ is not the zero function is at most $23$.

\subsubsection{Traders $T_i$, where $i\in [2n]$}

Finally, for each $i\in [2n]$, we set the initial endowment vector
  $\ww_i=(w_{i,1},...,w_{i,2n+2})$ of $T_i$ as follows:
$$
w_{i,2n+1}=1/n^{12}\ \ \ \ \text{and}\ \ \ \
w_{i,k}=0,\ \ \ \ \text{for all other $k\in [2n+2]$.}
$$
We set the PLC utility functions $r_{i,k}(\cdot)$, where $k\in [2n+2]$, of $T_i$ as follows:
$$r_{i,i}(\cdot)\Leftarrow [1]\ \ \ \ \text{and}\ \ \ \
  r_{i,k}(\cdot)\Leftarrow [0],\ \ \ \ \text{for all other $k\in [2n+2]$.}\vspace{0.15cm}$$



\subsection{From Approximate Market Equilibria to Approximate Nash Equilibria}

By definition, $\calM$ is a $2$-linear additively separable
  PLC market which is strongly connected, $27$-bounded and $23$-sparse.
Let $N=2n+2$, the number of goods in $\calM$.
To prove Theorem \ref{main}, we only need to show that from any
  $N^{-13}$-approximate market equilibrium $\pp$~of $\calM$, one can construct
  an $ n^{-6}$-well-supported Nash equilibrium $(\xx,\yy)$
  of $(\AA,\BB)$ in polynomial time.
To this end, let $(\xx',\yy')$ denote the following
  two $n$-dimensional vectors:\vspace{-0.08cm}
\begin{equation}\label{usefuleq1}
x'_k=p_k-1\ \ \ \ \text{and}\ \ \ \ y'_k=p_{n+k}-1,\ \ \ \ \text{for all $k\in [n]$.}
\end{equation}
Then, we normalize $(\xx',\yy')$ to get a pair of distributions
  $(\xx,\yy)$ (we will show later that $\xx',\yy'\ne \00$):\vspace{0.06cm}
\begin{equation}\label{usefuleq2}
x_k=\frac{x'_k}{\sum_{i\in [n]} x'_i}\ \ \ \ \text{and}\ \ \ \
y_k=\frac{y'_k}{\sum_{i\in [n]} y'_i},\ \ \ \ \ \text{for all $k\in [n]$.}\vspace{0.06cm}
\end{equation}
Theorem \ref{main} follows directly from Theorem \ref{mainhaha}
  which we will prove in the next section.
Note that if $\pp$
  is a $N^{-13}$-approximate equilibrium, then it is
  also an $n^{-13}$-approximate equilibrium by definition.

\begin{theo}\label{mainhaha}
If $\pp$ is an $n^{-13}$-approximate market equilibrium of $\calM$,
  then $(\xx,\yy)$ constructed above is an $n^{-6}$-well-supported
  Nash equilibrium of $(\AA,\BB)$.\vspace{0.05cm}
\end{theo}

\section{Correctness of the Reduction}

In this section, we prove Theorem \ref{mainhaha}.
Let $\pp=(p_1,...,p_{2n+2})$ be an normalized $n^{-13}$-approximate market equilibrium
  of $\calM$.
By the same argument used earlier, we can prove that $p_k>0$ for
  all $k\in [2n+2]$.
Therefore, we have $p_k\ge 1$ for all $k$ and $\min_k p_k=1$.
Let $\calX$ be an optimal allocation with respect to $\pp$
  that clears the market approximately:
$$
\calX=\Big\{\aa_\ss,\aa_\uu,\aa_\vv,\aa_i\in \mathbb{R}_+^{2n+2}:
  \ss\in S,\uu\in U,\vv\in V,i\in [2n]\Big\}.
$$

We start with the following notation.
Let $\calT'\subseteq \calT$ be a subset of traders, and $k\in [2n+2]$. We use
  $w_k[\calT']$ to denote the amount of good $G_k$ that traders in $\calT'$
  possess at the beginning and $a_k[\calT']$
  to denote the amount of good $G_k$ that $\calT'$ receives in the
  final allocation $\calX$.\vspace{0.01cm}

According to our construction, $w_k[\calT]\in [2,3]$ for every $k\in [2n+2]$.
Because $\calX$ clears the market approximately, we have
\begin{equation}\label{usefuleq3}
\big|w_k[\calT]-a_k[\calT]\big|\le w_k[\calT]/n^{13}\le 3/n^{13},
\ \ \ \ \ \text{for all $k\in [2n+2]$.}
\end{equation}
We further divide the traders into two groups:
$
\calT_1=\{T_\ss:\ss\in S\}$  and $\calT_2=\calT-\calT_1.
$
Then (\ref{usefuleq3}) implies
\begin{equation}\label{usefuleq4}
\big| w_k[\calT_1]-a_k[\calT_1]
  + w_k[\calT_2]-a_k[\calT_2] \big| \le 3/n^{13},
\ \ \ \ \ \text{for all $k\in [2n+2]$.}\vspace{0.06cm}
\end{equation}

\subsection{The Price-Regulation Property}

First, we show that, the price vector $\pp$ must still
  satisfy the price-regulation property as in the
  price-regulating market $\calM_{2n+2}$.
We will use the fact that traders in $\calT_1$ possess almost
  all the goods in $\calM$.\vspace{0.06cm}

\begin{lemm}[Price Regulation]\label{bounded}
For all $k\in [2n+2]$, $1\le p_k\le 2$.\vspace{0.06cm}
\end{lemm}
\begin{proof}
Assume for contradiction that $\pp$ does not satisfies the
  price-regulation property.
Then without loss of generality, we assume that $p_1=\max_k p_k>2$ and $p_2=1$.

By the same argument used in the proof of Property \ref{simplelemma},
  we have
$$
w_1[\calT_1]=(2n+1)\cdot \frac{1}{n},\ \ \ a_1[\calT_1]\le
  2n\cdot \frac{1}{n},\ \ \ \text{and thus,\ \ \ }
  w_1[\calT_1]-a_1[\calT_1]\ge \frac{1}{n}.
$$
By (\ref{usefuleq4}), we have
\begin{equation}\label{eq5}
w_1[\calT_2]-a_1[\calT_2]\le -\frac{1}{n}+\frac{3}{n^{13}}\ \ \Longrightarrow\ \
a_1[\calT_2]\ge w_1[\calT_2]+\frac{1}{n}-\frac{3}{n^{13}}\ge
\frac{1}{n}-\frac{3}{n^{13}}
\end{equation}
because $w_1[\calT_2]\ge 0$.
However, this cannot be true since the amount of goods the
  traders in $\calT_2$ possess at the beginning is much smaller compared to $1/n$.
Even if they spend all the money on $G_1$, we still have
$$
a_1[\calT_2]\le \frac{\sum_{k\in [2n+2]}p_k\cdot w_k[\calT_2]}{p_1}
  \le \sum_{k\in [2n+2]} w_k[\calT_2]=O(n^{-2})\ll \frac{1}{n},
$$
since we assumed that $p_1=\max_k p_k$. This contradicts with (\ref{eq5}).\vspace{0.05cm}
\end{proof}




\subsection{Relations between $p_k$ and $w_k[\calT_2]-a_k[\calT_2]$}

Next, we prove two very useful relations between $p_k$ and
  $w_k[\calT_2]-a_k[\calT_2]$.\vspace{0.05cm}

\begin{lemm}\label{kaka1}
Let $\pp$ be a normalized $n^{-13}$-approximate market equilibrium
  and $\calX$ be an optimal allocation that clears the market approximately.
If $w_k[\calT_2]-a_k[\calT_2]>3/n^{13}$ for some $k\in [2n+2]$, then $p_k=1$.\vspace{0.05cm}
\end{lemm}
\begin{proof}
Without loss of generality, we prove the lemma for the case when $k=1$.
By (\ref{usefuleq4}), we have \vspace{-0.05cm}
\begin{equation*}
w_1[\calT_1]-a_1[\calT_1]<0.\vspace{-0.05cm}
\end{equation*}
This means that, in the market participated
  by traders $T_\ss$,
  the amount of $G_1$ which they would like to buy is strictly more than
  the amount of $G_1$ they possess at the beginning.
Intuitively this implies that the price $p_1$ of $G_1$ is lower than what it should be,
  and indeed we show below that $p_1=\min_k p_k=1$.

On one hand, by the construction, only the following traders $T_\ss$
  are interested in $G_1$:
$$
S_1=\{\ss=(1,j):j\ne 1\}\ \ \ \ \text{and}\ \ \ \ S_2=\{\ss=(i,1):i\ne 1\}.
$$
On the other hand, we have $$a_1[T_\ss,\ss\in S_1]\le w_1[T_\ss,\ss\in S_1]=w_1[\calT_1]$$
  due to the budget limitation.
As a result, there must exist an $\ss=(i,1)\in S_2$ such that
$
a_{\ss,1}>0.
$
Since $\aa_\ss$ is an optimal bundle for $T_\ss$ with respect to $\pp$, we have
$$
\frac{1}{p_1}\ge \frac{2}{p_i}\ \ \Longrightarrow\ \ p_1\le \frac{p_i}{2}.
$$
By Lemma \ref{bounded}, the price-regulation property, we
  conclude that $p_1=1$ and the lemma is proved.\vspace{0.05cm}
\end{proof}

\begin{lemm}\label{kaka2}
Let $\pp$ be a normalized $n^{-13}$-approximate market equilibrium
  and $\calX$ be an optimal allocation that clears the market approximately.
If $w_k[\calT_2]-a_k[\calT_2]<-3/n^{13}$ for some $k\in [2n+2]$, then $p_k=2$.\vspace{0.05cm}
\end{lemm}
\begin{proof}
Without loss of generality, we prove the lemma for the case when $k=1$.
By (\ref{usefuleq4}), we have\vspace{-0.05cm}
\begin{equation*}
w_1[\calT_1]-a_1[\calT_1]>0.\vspace{-0.05cm}
\end{equation*}
This means that, in the market participated
  by traders $T_\ss$,
  the amount of $G_1$ which they would like to buy is strictly less than
  the amount of $G_1$ they possess at the beginning.
Intuitively, this implies that the~price $p_1$ of $G_1$ is higher than what it should be,
  and indeed we show below that $p_1=2=\max_k p_k$.

Since $a_1[\calT_1]<w_1[\calT_1]$, there must exist a $j\in [2n+2]$
  with $j\ne 1$ such that $\ss=(1,j)$ and\vspace{-0.05cm}
$$
a_{\ss,1}<w_{\ss,1}.\vspace{-0.05cm}
$$
(Otherwise $a_1[\calT_1]\ge w_1[\calT_1]$).
This means that $T_{\ss}$ spends some of its money to buy $G_j$ and thus,
$$
\frac{1}{p_j}\ge \frac{2}{p_1}\ \ \Longrightarrow\ \ p_1\ge 2p_j.
$$
By Lemma \ref{bounded}, the price-regulation property,
  we conclude that $p_1=2$ and the lemma is proved.
\end{proof}

We also need the following two lemmas. We only prove the first one.
The second one can be proved symmetrically.

\begin{lemm}\label{hahalemma1}
Let $\uu=(i,j,1)$ be a triple in $U$ and
  $\uu'=(j,i,1)\in U$.
Then for any $k\in [2n+1]$, we have\vspace{-0.08cm}
\begin{equation}\label{usefuleq6}
w_{\uu,k}+w_{\uu',k}\ge a_{\uu,k}+a_{\uu',k}.\vspace{0.08cm}
\end{equation}
\end{lemm}

\begin{lemm}\label{hahalemma2}
Let $\vv=(i,j,2)$ be a triple in $V$ and
  $\vv'=(j,i,2)\in V$.
Then for any $k\in [2n+1]$, we have\vspace{-0.08cm}
$$
w_{\vv,k}+w_{\vv',k}\ge a_{\vv,k}+a_{\vv',k}.\vspace{0.1cm}
$$
\end{lemm}

\begin{proof}[Proof of Lemma \ref{hahalemma1}]
Without loss of generality, we only need to prove Lemma \ref{hahalemma1}
  for the case when $\uu=(1,2,1)$ and $\uu'=(2,1,1)$.
Let $\CC$ and $\DD$ denote the following two $n$-dimensional vectors: For $k\in [n]$,
\begin{equation}\label{yaya1}
(C_k,D_k)=(A_{1,k}-A_{2,k},0)\ \text{if}\ A_{1,k}-A_{2,k}\ge 0;\ \text{and}\
(C_k,D_k)=(0,A_{2,k}-A_{1,k})\ \text{otherwise}.
\end{equation}
We also define $E$ and $F$ to be the following two nonnegative numbers:\vspace{0.04cm}
\begin{equation}\label{yaya2}
(E,F)=\left(\sum_{k\in [n]}D_k-\sum_{k\in [n]}C_k,0\right)\ \text{if}\ \sum_{k\in [n]}
  D_k\ge \sum_{k\in [n]} C_k;
\ (E,F)=\left(0,\sum_{k\in [n]}C_k-\sum_{k\in [n]}D_k\right)\ \text{otherwise}.\vspace{0.06cm}
\end{equation}
Then by the construction, we have $w_{\uu,n+k}=C_k/n^5$ and
  $w_{\uu',n+k}=D_k/n^5$ for all $k\in [n]$,
$$w_{\uu,1}= w_{\uu',2}=1/n^4,\ \ \ w_{\uu,2n+1}=E/n^5,\ \ \
  w_{\uu',2n+1}=F/n^5,$$ and all other entries of $\ww_\uu$
  and $\ww_{\uu'}$ are $0$.

We now focus on the preference of $T_\uu$.
After selling its initial endowment, the budget of $T_\uu$ is
$$
p_1\cdot \frac{1}{n^4}+\sum_{k\in [n]}
  p_{n+k}\cdot \frac{C_k}{n^5}+p_{2n+1}\cdot \frac{E}{n^5}=\Omega\left(\frac{1}{n^4}\right)
$$
by Lemma \ref{bounded}.
The PLC utility functions $r_{\uu,k}(\cdot)$ of $T_\uu$ are designed carefully,
  so that even though we do not know what exactly $\pp$ is,
  we know the behavior of $T_\uu$ due to the price-regulation property:
$T_\uu$ first buys the following bundle of goods from the market
\begin{equation}\label{bundle}
\Big\{\frac{D_k}{n^5}\ \text{amount of $G_{n+k}$ and $\frac{F}{n^5}$
  amount of $G_{2n+1}$}:k\in [n]\Big\}.
\end{equation}
As $\DD$ has at most $20$ nonzero entries and every entry
  is between $0$ and $2$, the cost of this bundle is
$$
\sum_{k\in [n]} p_{n+k}\cdot \frac{D_k}{n^5}+p_{2n+1}\cdot
  \frac{F}{n^5}=O\left(\frac{1}{n^5}\right)\ll \frac{1}{n^4}.
$$
$T_\uu$ then buys as much $G_1$ as it can up to $1/n^4$, and
  spends all the money left, if any, on $G_{2n+2}$.\vspace{0.01cm}

The behavior of $T_{\uu'}$ is similar.
It first buys the following bundle of goods from the market:
\begin{equation}\label{bundle2}
\Big\{\frac{C_k}{n^5}\ \text{amount of $G_{n+k}$ and $\frac{E}{n^5}$
  amount of $G_{2n+1}$}:k\in [n]\Big\}.
\end{equation}
It then buys as much $G_2$ as it can up to $1/n^4$, and spends all
  the money left, if any, on $G_{2n+2}$.\vspace{0.01cm}

Now we are ready to prove the lemma. The case when $k\in [n]$ but $k\ne 1,2$ is
  trivial since $$w_{\uu,k}=w_{\uu',k}=a_{\uu,k}=a_{\uu',k}=0.$$
When $k=1$, we have $w_{\uu,1}+w_{\uu',1}=1/n^4$,
  $a_{\uu',1}=0$, $a_{\uu,1}\le 1/n^4$ and thus, (\ref{usefuleq6}) follows.
The case when $k=2$ can be proved similarly.
For the case of $n+k$ where $k\in [n]$, we have
$$
w_{\uu,n+k}=\frac{C_k}{n^5},\ \ \ w_{\uu',n+k}=\frac{D_k}{n^5},\ \ \
a_{\uu,n+k}=\frac{D_k}{n^5},\ \ \ \text{and}\ \ \ a_{\uu',n+k}=\frac{C_k}{n^5},
$$
and (\ref{usefuleq6}) follows. When $k=2n+1$, we have
$$
w_{\uu,2n+1}=\frac{E}{n^5},\ \ \ w_{\uu',2n+1}=\frac{F}{n^5},\ \ \
a_{\uu,2n+1}=\frac{F}{n^5},\ \ \ \text{and}\ \ \ a_{\uu',2n+1}=\frac{E}{n^5},
$$
and (\ref{usefuleq6}) follows. This finishes the proof of the lemma.\vspace{0.05cm}
\end{proof}

By Lemma \ref{hahalemma1}, Lemma \ref{hahalemma2} and Lemma \ref{kaka1},
  we immediately get the following corollary concerning $p_{2n+1}$.\vspace{0.05cm}

\begin{coro}\label{coro1}
$p_{2n+1}=1$.
\end{coro}
\begin{proof}
First, by Lemma \ref{hahalemma1} and Lemma \ref{hahalemma2},
  we have
$$
w_{2n+1}[T_\uu,T_\vv:\uu\in U,\vv\in V]-a_{2n+1}[T_\uu,T_\vv:\uu\in U,\vv\in V]\ge 0.
$$
However, the construction implies that
$$
w_{2n+1}\big[T_i:i\in [2n]\big]=2n\cdot \frac{1}{n^{12}}=
\frac{2}{n^{11}}\ \ \ \ \text{and}\ \ \ \ a_{2n+1}\big[T_i:i\in [2n]\big]=0.
$$
As a result, $w_{2n+1}[\calT_2]-a_{2n+1}[\calT_2]\ge 2/n^{11}\gg 3/n^{13}.$
It then follows from Lemma \ref{kaka1} that $p_{2n+1}=1$.
\end{proof}

\subsection{Proof of Theorem \ref{mainhaha}}

Now we let $\xx'$ and $\yy'$ denote the two vectors obtained
  in (\ref{usefuleq1}).
By Lemma \ref{bounded} we have $x_k',y_k'\in [0,1]$ for all $k\in [n]$.
We will prove the following two properties of $(\xx',\yy')$
  and use them to prove Theorem \ref{mainhaha}.\vspace{0.05cm}

\begin{prope}\label{finalprope1}
For all $1\le i\ne j\le n$, we have
\begin{eqnarray}\label{lateruse}
&(\AA_i-\AA_j)\yy'^T< -\epsilon\ \ \Longrightarrow\ \ x_i'=0;&\text{and}\\[0.3ex]
\label{uselater}&\xx'(\BB_i-\BB_j) < -\epsilon\ \ \Longrightarrow\ \ y_i'=0,
\end{eqnarray}
where $\epsilon=n^{-6}$, $\AA_i$ denotes the $i$\emph{th} row vector of $\AA$,
  and $\BB_i$ denotes the $i$\emph{th} column vector of $\BB$.\vspace{0.05cm}
\end{prope}

\begin{prope}\label{finalprope2}
There exist $i$ and $j\in [n]$ such that $x_i'=1$ and $y_j'=1$.\vspace{0.05cm}
\end{prope}

Now assume that $\xx'$ and $\yy'$ satisfy both properties. In particular,
  Property \ref{finalprope2} implies that
  $\xx',\yy'\ne \00$. As a result,
we can normalize them to get two probability distribution
  $\xx$ and $\yy$ using (\ref{usefuleq2}).
Before proving these two properties, we show that
  $(\xx,\yy)$ must be an $\epsilon$-well-supported Nash equilibrium of $(\AA,\BB)$.\vspace{0.05cm}

\begin{proof}[Proof of Theorem \ref{mainhaha}]
Since both $\xx$ and $\yy$ are probability distributions, we only need
  to show that $(\xx,\yy)$ satisfies (\ref{nasheq1}) and (\ref{nasheq2})
  for all $i,j:1\le i\ne j\le n$.
We only prove (\ref{nasheq1}) here.\vspace{0.01cm}

Assume $\AA_i\yy^T+\epsilon <\AA_j\yy^T$, then we have\vspace{-0.1cm}
$$
(\AA_i-\AA_j)\yy'^T= (\AA_i-\AA_j)\yy^T\cdot \left(\sum_{k\in [n]} y_k'\right)
<-\epsilon
$$
since $\sum_{k\in [n]}y_k'\ge 1$ by Property \ref{finalprope2}.
As a result, by Property \ref{finalprope1} we have $x_i'=0$ and thus, $x_i=0$.\vspace{0.04cm}
\end{proof}

Finally, we prove Property \ref{finalprope1} and Property \ref{finalprope2}.\vspace{0.04cm}

\begin{proof}[Proof of Property \ref{finalprope1}]
We only prove (\ref{lateruse}) for the case when $i=1$, $j=2$.
(\ref{uselater}) can be proved similarly.\vspace{0.01cm}

Let $\uu=(1,2,1)$ and $\uu'=(2,1,1)$.
Let $\CC$ and $\DD$ be the two nonnegative vectors defined in (\ref{yaya1}),
  and $E$ and $F$ be the two nonnegative numbers defined in (\ref{yaya2}).
We have
\begin{equation}\label{yuyu1}
\CC-\DD=\AA_1-\AA_2\ \ \ \ \text{and}\ \ \ \ E+\sum_{k\in [n]}C_k=F+\sum_{k\in [n]}D_k.
\end{equation}

Assume $(\AA_1-\AA_2)\yy'^T<-\epsilon$.
Then the money of $T_\uu$ left after purchasing the bundle in (\ref{bundle}) is
$$
p_1\cdot \frac{1}{n^4}+\sum_{k\in [n]} p_{n+k}\cdot \frac{C_k}{n^5}+p_{2n+1}\cdot \frac{E}{n^5}
  -\sum_{k\in [n]}p_{n+k}\cdot \frac{D_k}{n^5}-p_{2n+1}\cdot \frac{F}{n^5}.
$$
By Corollary \ref{coro1}, we have $p_{2n+1}=1$. Using (\ref{yuyu1}), we can simplify the
  equation to be the following:\vspace{0.04cm}
\begin{equation}\label{rrr}
p_1\cdot \frac{1}{n^4}+\frac{1}{n^5}\sum_{k\in [n]}y_k'\cdot (C_k-D_k)
=p_1\cdot \frac{1}{n^4}+\frac{1}{n^5}(\AA_1-\AA_2)\yy'^T<p_1\cdot \frac{1}{n^4}-
  \frac{\epsilon}{n^5}.
\end{equation}
This implies that the amount $a_{\uu,1}$ of $G_1$ that $T_\uu$
  buys is smaller than
$$
\frac{1}{n^4}-\frac{\epsilon}{p_1n^5}\le \frac{1}{n^4}-\frac{1}{2n^{11}}
$$
since $\epsilon=n^{-6}$.
However, we have $w_{\uu,1}=1/n^4$ and thus,
\begin{equation}\label{summary}
w_{\uu,1}-a_{\uu,1}>1/(2n^{11}).
\end{equation}

On the other hand, it is easy to check that $w_{\uu',1}=0$ and $a_{\uu',1}=0$.
By Lemma \ref{hahalemma1} and \ref{hahalemma2},
  we have
\begin{equation}\label{lastequation}
w_1[T_\uu,T_\vv:\uu\in U,\vv\in V]-a_1[T_\uu,T_\vv:\uu\in U,\vv\in V]>\frac{1}{2n^{11}}.
\end{equation}
Next we bound $w_1\big[T_i:i\in [2n]\big]-a_1\big[T_i:i\in [2n]\big]$.
By the construction, $a_1\big[T_i:i\in [2n],i\ne 1\big]=0$ and
$$
a_{1,1}=\frac{p_{2n+1}\cdot \frac{1}{n^{12}}}{p_1}\le \frac{1}{n^{12}},
$$
since $p_{2n+1}=1$. Therefore, $w_1\big[T_i:i\in [2n]\big]-a_1\big[T_i:i\in [2n]\big]\ge -1/n^{12}$.
Combining (\ref{lastequation}), we have
$$
w_1[\calT_2]-a_1[\calT_2]>\frac{1}{2n^{11}}-\frac{1}{n^{12}}\gg \frac{3}{n^{13}}.
$$
It then follows from Lemma \ref{kaka1} that $p_1=1$ and thus, $x_1'=0$.\vspace{0.06cm}
\end{proof}

\begin{proof}[Proof of Property \ref{finalprope2}]
Let $\ell\in [n]$ be one of the indices that maximizes $\AA_\ell\yy'^T$,
then we show that $x_\ell'=1$. Without loss of generality, we may assume that $\ell=1$.

First, we consider $\vv=(i,j,2)$ and $\vv'=(j,i,2)$ in $V$.
In the proof of Lemma \ref{hahalemma1}, we showed that\vspace{-0.06cm}
$$w_{\uu,n+k}+w_{\uu',n+k}=a_{\uu,n+k}+a_{\uu',n+k},\vspace{-0.06cm}$$
{for all pairs $\uu=(i,j,1)$ and $\uu'=(j,i,1)$, and all $k\in [n]$.}
Similarly, we can prove that\vspace{-0.06cm}
\begin{equation}\label{symmetry}
w_{\vv,1}+w_{\vv',1}=a_{\vv,1}+a_{\vv',1}.\vspace{-0.06cm}
\end{equation}

Second, for every $\uu=(i,j,1)\in U$, we always have $w_{\uu,1}=a_{\uu,1}$.
  This is because
\begin{enumerate}
\item If $i\ne 1$, then $w_{\uu,1}=a_{\uu,1}=0$; and
\item If $i=1$, then by (\ref{rrr}), the money of $T_\uu$ left after
  purchasing the bundle of goods in (\ref{bundle})
  is at\\ least $p_1/n^4$, so $w_{\uu,1}=a_{\uu,1}=1/n^4$.
\end{enumerate}
As a result, we have $w_1[T_\uu,T_\vv:\uu\in U,\vv\in V]=a_1[T_\uu,T_\vv:\uu\in U,\vv\in V]$.

However, the amount of $G_1$ that $T_1$  buys is\vspace{-0.05cm}
$$
\frac{p_{2n+1}\cdot \frac{1}{n^{12}}}{p_1}\ge \frac{1}{2n^{12}}
$$
and thus, $w_1\big[T_i,i\in [2n]\big]-a_1\big[T_i:i\in [2n]\big]\le -1/(2n^{12})$.
Putting everything together, we have
$$w_1[\calT_2]-a_1[\calT_2]\le -\frac{1}{2n^{12}}\ll -\frac{3}{n^{13}}.$$
By Lemma \ref{kaka2}, we conclude that $p_1=2$ and thus, $x_1'=1$.\vspace{0.2cm}
\end{proof}

\bibliographystyle{plain}
\bibliography{MARKET}

\appendix
\section{Proof of Theorem \ref{existence}}\label{appendix1}

In this section, we prove Theorem \ref{existence}.
To this end, we first show that under the conditions of
  Theorem \ref{existence},
  $\calM$ has at least one \emph{quasi-equilibrium} (see the
  definition below).
Then we show that any quasi-equilibrium of $\calM$ is
  indeed a market equilibrium.\vspace{0.06cm}

\begin{defi} \label{quasi-equilibrium}
A quasi-equilibrium of $\calM$ is a \emph{(normalized)} price vector
  $\pp\in \mathbb{R}_+^n$ such that there
  exists an allocation $\mathcal{X}=\{\xx_i\in \mathbb{R}_+^n:i\in [m]\}$
  which has the following properties:
\begin{enumerate}
\item The market clears: For every good $G_j \in \mathcal{G}$,
\begin{equation*}
\sum_{i \in [m]} x_{i,j}\le \sum_{i \in [m]} w_{i,j};
\end{equation*}
In particular, if $p_j>0$, then\vspace{-0.1cm}
\begin{equation*}
\sum_{i \in [m]} x_{i,j}=\sum_{i \in [m]} w_{i,j};
\end{equation*}
\item For every trader $T_i\in \calT$, at least one of the
  following is true:
\begin{enumerate}
\item
  $\xx_i \in \text{\sc OPT}(i,\pp)$;
\item $\pp \cdot \xx_i =
  \pp \cdot \ww_{i}=0$ \emph{(}zero income\emph{)}.\vspace{0.06cm}
\end{enumerate}
\end{enumerate}
\end{defi}

The difference between market equilibria and quasi-equilibria
  is that in the latter, we \emph{do not} require the optimality of
  allocations for traders with a zero income: If a trader
  has a zero income, then we can assign her any bundle of zero cost.
However, if $\pp$ is a quasi-equilibrium and
  the income of every trader is positive with respect to $\pp$, then
  by definition $\pp$ must be a market equilibrium.\vspace{0.01cm}

In \cite{Maxfield} Maxfield gave a set of conditions that are sufficient
  for the existence of a quasi-equilibrium in an exchange market.
We use the following simplified version \cite{Maxfield}:\vspace{0.04cm}

\begin{theo}[\cite{Maxfield}] \label{quasi-existence}
An exchange market $\calM$ has a quasi-equilibrium $\pp$ if \begin{enumerate}
\item For each trader $T_i\in \calT$,
  its utility function $u_i:\mathbb{R}_+^n\rightarrow \mathbb{R}$
  is both continuous and quasi-concave; and\vspace{-0.08cm}
\item For each trader $T_i \in \calT$,
  $u_i$ is non-satiable, i.e., for any $\xx \in
  \mathbb{R}_+^n$, there exists a vector $\yy\in \mathbb{R}_+^n$ such that
  $u_i(\yy) > u_i(\xx)$.\vspace{0.04cm}
\end{enumerate}
\end{theo}

Now we use Theorem \ref{quasi-existence} to prove Theorem \ref{existence}.

\begin{proof}[Proof of Theorem \ref{existence}]
First, it is easy to check that if $\calM$ is an additively separable
  PLC market that is strongly connected, then it satisfies both conditions
  in Theorem \ref{quasi-existence}.
In particular, $u_i$ is non-satiable since the economy graph of $\calM$ is
  strongly connected and thus, there exists a $j\in [n]$
  such that $r_{i,j}(\cdot)$ is strictly monotone.
As a result, $\calM$ must have a quasi-equilibrium $\pp$.
We use $\calX=\{\xx_i\in \mathbb{R}_+^n:i\in [m]\}$ to
  denote an allocation that clears the market.
Since $\pp\ne \00$, there is at least one trader in $\calT$,
  say $T_1\in \calT$, has a positive income.

Second, we show that for every trader, its income is positive and
  thus, $\pp$ is indeed an equilibrium of $\calM$.
Suppose this is not true, then there is at least one trader $T_2$
  whose income is zero.
Since the economy graph is strongly connected,
  there is a directed path from $T_2$ to $T_1$.
As a result, there must be a directed edge $T_3T_4$
  on the path such that the income of $T_3$ is zero and
  the income of $T_4$ is positive.
By definition, there exists a $j\in [n]$ such that the amount
  of $G_j$ that $T_3$ owns at the beginning is positive
  and the PLC utility function of $T_4$ with respect to $G_j$ is strictly monotone.
However, since the income of $T_3$ is zero, we have $p_j=0$ and
  thus, the amount of $G_j$ that $T_4$ wants to buy is $+\infty$,
  contradicting the assumption that $\pp$ is a quasi-equilibrium of $\calM$
  (since the income of $T_4$ is positive but the bundle she
  receives is not optimal).\vspace{0.01cm}


Now we have proved the existence of a market equilibrium $\pp$.
The second part of Theorem \ref{existence} follows from the work of
  Devanur and Kannan \cite{Nikhil}.
In \cite{Nikhil}, the authors proposed an algorithm for computing a market
  equilibrium in an additively separable PLC market\hspace{0.04cm}\footnote{When
  the number of goods is constant, the algorithm is polynomial-time.}.
They divide the whole search space $\mathbb{R}_+^n$ of $\pp$ into ``cells''
  $C\subset \mathbb{R}_+^n$ using hyperplanes.
Then for each cell $C$, there is a rational linear program LP$_C$ that
  characterizes the set of market equilibria in $C$:
  $\pp\in C$ is an equilibrium of $\calM$ if and only if
  it is a feasible solution to LP$_C$ (In particular, if LP$_C$ has no feasible solution then there
  is no equilibrium in $C$).
Moreover, the size of LP$_C$, for any cell $C$,
  is polynomial in the size of $\calM$.\vspace{0.01cm}

Now let $\pp$ be a market equilibrium of $\calM$, which
  is not necessarily rational.
We let $C^*$ denote the cell that $\pp$ lies in, then
  $\pp$ must be a feasible solution to LP$_{C^*}$.
Since LP$_{C^*}$ is rational, it must have a rational solution
  $\pp^*$ and the number of bits one need to describe $\pp^*$
  is polynomial in the size of LP$_{C^*}$ and thus, is polynomial
  in the size of $\calM$.
Theorem \ref{existence} then follows since $\pp^*$ is also an equilibrium of $\calM$.
\end{proof}
\end{document}